\documentclass[copyright]{eptcs}

\usepackage{graphicx}
\usepackage[latin1]{inputenc}
\usepackage{amscd, amsfonts, amssymb, amsmath, amsthm}
\usepackage[curve]{xypic}
\usepackage{tikz}

\newcommand{\K}{\mathit K}
\newcommand{\M}{\mathit M}

\newcommand{\nxt}{\mathbf X}
\newcommand{\fut}{\mathbf F}
\newcommand{\gbl}{\mathbf G}
\newcommand{\unt}{\mathbf U}
\newcommand{\rel}{\mathbf R}

\newcommand{\equ}{\; \Longleftrightarrow \;}

\newcommand{\la}{\langle}
\newcommand{\ra}{\rangle}

\newcommand{\et}{\mbox{ and }}
\newcommand{\st}{\mbox{ s.t. }}

\newcommand{\imp}{\; \Longrightarrow \;}
\newcommand{\sep}{\; | \;}

\newcommand{\Q}{\mathit Q}
\newfont{\amstoto}{msbm10}
\newcommand{\NN}{\mbox{\amstoto\char'116}}

\newcommand{\TF}{{\cal T(F)}}
\newcommand{\TFX}{{\cal T(F, X)}}
\newcommand{\TFQ}{{\cal T(F \cup \Q)}}

\newcommand{\Qf}{\Q_F}

\newcommand{\var}{{\cal V}ar}
\newcommand{\pos}{{\cal P}os}
\newcommand{\R}{\mathcal R}
\newcommand{\desc}{\R^*}
\newcommand{\T}{\mathcal T}
\newcommand{\X}{\mathcal X}
\newcommand{\vars}{\mathcal V}
\newcommand{\rw}{\rightarrow}
\newcommand{\trw}{\rightarrow^{\lambda}}
\newcommand{\lrw}{\longrightarrow}

\newcommand{\arw}{\dashrightarrow}

\newcommand{\rwne}{\rightarrow^{\not\varepsilon}}

\newcommand{\Deps}{\Delta_{\varepsilon}}
\newcommand{\norm}{\downarrow}

\newcommand{\A}{\mathit A}
\newcommand{\B}{\mathit B}

\newcommand{\F}{\mathcal F}

\newcommand{\Lang}[1]{\mathcal{L}_{#1}}
\newcommand{\f}{\rw}
\newcommand{\aaex}{{\it A}_{\R}}

\usetikzlibrary{automata}

\newtheorem{theorem}{Theorem}
\newtheorem{lemma}{Lemma}
\newtheorem{definition}{Definition}
\newtheorem{property}{Property}
\newtheorem{example}{Example}

\newcommand{\titre}{Verifying Temporal Regular Properties of
  Abstractions of Term Rewriting Systems}

\title{\titre}
 
\author{Benoît Boyer
\institute{Université Rennes 1, France}
\email{Benoit.Boyer@irisa.fr}
\and
Thomas Genet
\institute{Université Rennes 1, France}
\email{Thomas.Genet@irisa.fr}
}
\def\titlerunning{\titre}
\def\authorrunning{Boyer \& Genet}
\begin{document}
\maketitle

\begin{abstract}
  The tree automaton completion is an algorithm used for proving safety
  properties of systems that can be modeled by a term rewriting system. This
  representation and verification technique works well for proving properties of
  infinite systems like cryptographic protocols or more recently on Java
  Bytecode programs. This algorithm computes a tree automaton which represents
  a (regular) over approximation of the set of reachable terms by rewriting
  initial terms. This approach is limited by the lack of information about
  rewriting relation between terms. Actually, terms in relation by rewriting are
  in the same equivalence class: there are recognized by the same state in the
  tree automaton.

  Our objective is to produce an automaton embedding an 
  abstraction of the rewriting relation sufficient to prove temporal
  properties of the term rewriting system.

  We propose to extend the algorithm to produce an automaton having more
  equivalence classes to distinguish a term or a subterm from its successors w.r.t. rewriting. 
  While ground transitions are used to recognize equivalence classes of terms,
  $\epsilon$-transitions represent the rewriting relation between terms.
  From the completed automaton, it is possible to automatically build a
  Kripke structure abstracting the rewriting sequence.
  States of the Kripke structure are states of the tree automaton and the
  transition relation is given by the set of $\epsilon$-transitions.
  States of the Kripke structure are labelled by the set of terms recognized
  using ground transitions. On this Kripke structure, we define the Regular
  Linear Temporal Logic (R-LTL) for expressing properties. Such properties can then 
  be checked using standard model checking algorithms. The only
  difference between LTL and R-LTL is that predicates are replaced by
  regular sets of acceptable terms.  
\end{abstract}

\section{Introduction}


Our main objective is to formally verify programs or systems modeled
using Term Rewriting Systems.  In a previous
work~\cite{BoichutGJL-RTA07}, we have shown that it is possible to
translate a Java bytecode program into a Term Rewriting System (TRS).
In this case, terms model Java Virtual Machine (JVM) states and the
execution of bytecode instructions is represented by rewriting,
according to the small-step semantics of Java. An interesting point of
this approach is the possibility to classify rewriting rules. More
precisely, there is a strong relation between the position of
rewriting in a term and the semantics of the executed transition on
the corresponding state. For the case of Java bytecode, since a term
represents a JVM state, rewriting at the top-most position corresponds
to manipulations of the call stack, i.e. it simulates a method call or
method return.  On the other hand, since the left-most subterm
represents the execution context of the current method (so called
frame), rewriting at this position simulates the execution of the code
of {\em this} method. Hence, by focusing on rewriting at a particular
position, it is possible to analyse a Java program at the method call
level (inter procedural control flow) or at the instruction level
(local control flow).
The contribution of this paper is dual. First, we propose an abstract rewriting
relation to characterize the rewriting paths at a particular depth in terms.
Second, we propose an algorithm which builds a tree automaton recognizing this
relation between terms. Thus, it is possible for instance to build
a tree automaton recognizing the graph of method calls by abstracting the
rewriting relation for the top-most position of JVM terms.

The verification technique used in~\cite{BoichutGJL-RTA07}, called Tree Automata
Completion~\cite{FeuilladeGVTT-JAR04}, is able to finitely over-approximate the
set of reachable terms, i.e. the set of all reachable states of the
JVM. However, this technique lacks precision in the sense that it makes no
difference between all those reachable terms. Due to the approximation
algorithm, all reachable terms are considered as equivalent and the execution
ordering is lost. In particular, this prevents to prove temporal properties of such models. 
However, using approximations makes it possible to prove unreachability
properties of infinite state systems.

In this preliminary work, we propose to improve the Tree Automata Completion
method so as to prove temporal properties of a TRS representing a finite state
system. The first step is to refine the algorithm so as to produce a tree
automaton keeping an approximation of the rewriting relation between
terms. Then, in a second step, we propose a way to check LTL-like formulas on
this tree automaton.

\section{Preliminaries}

Comprehensive surveys can be found in~\cite{BaaderN-book98} for
rewriting, and in~\cite{TATA,GilleronTison-FI95} for tree automata
and tree language theory.

Let $\F$ be a finite set of symbols, each associated with an arity function, and
let $\X$ be a countable set of variables. $\TFX$ denotes the set of terms, and
$\TF$ denotes the set of ground terms (terms without variables). The set of
variables of a term $t$ is denoted by $\var(t)$. A substitution is a function
$\sigma$ from $\X$ into $\TFX$, which can be uniquely extended to an
endomorphism of $\TFX$. A position $p$ for a term $t$ is a word over $\NN$. The
empty sequence $\lambda$ denotes the top-most position. The set $\pos(t)$ of
positions of a term $t$ is inductively defined by:
\begin{itemize}
\item $\pos(t)= \{ \lambda\} $ if $t \in \X$
\item $\pos(f(t_1,\dots,t_n)) = \{ \lambda \} \cup \{i.p \mid 1 \leq i \leq n
  \et p \in \pos(t_i) \}$
\end{itemize}
If $p \in \pos(t)$, then $t|_p$ denotes the subterm of $t$ at position $p$ and
$t[s]_p$ denotes the term obtained by replacement of the subterm $t|_p$ at
position $p$ by the term $s$. A term rewriting system (TRS) $\R$ is a set of {\em
  rewrite rules} $l \rw r$, where $l, r \in \TFX$, $l \not \in \X$, and $\var(l)
\supseteq \var(r)$.
The TRS $\R$ induces a rewriting relation $\rw_{\R}$ on terms as follows. Let
$s, t\in \TFX$ and $l \rw r \in \R$, $s \rw^p_{\R} t$ denotes that there exists a
position $p\in\pos(t)$ and a substitution $\sigma$ such that $s|_p= l\sigma$ and
$r=s[r\sigma]_p$. Note that the rewriting position $p$ can generally be omitted,
i.e. we generally write $s \rw_{\R} t$. The reflexive transitive closure of
$\rw_{\R}$ is denoted by $\rw^*_{\R}$. The set 
of $\R$-descendants of a set of ground terms $E$ is $\desc(E) = \{t
\in \TF \sep \exists s \in E \st s \rw^*_{\R} t \}$.

The {\em verification technique} defined
in~\cite{Genet-RTA98,FeuilladeGVTT-JAR04} is based on the approximation of $\desc(E)$.
Note that $\desc(E)$ is possibly infinite: $\R$ may not terminate
and/or $E$ may be infinite. The set $\desc(E)$ is generally not
computable~\cite{GilleronTison-FI95}. However, it is possible to
over-approximate it~\cite{Genet-RTA98,FeuilladeGVTT-JAR04,Takai-RTA04}
using tree automata, i.e. a finite representation of infinite
(regular) sets of terms.  In this verification setting, the TRS $\R$
represents the system to verify, sets of terms $E$ and $Bad$ respectively 
represent the set of initial configurations and the set of ``bad''
configurations that should not be reached. Using tree automata
completion, we construct a tree automaton $\B$ whose language
$\Lang{}(\B)$ is such that $\Lang{}(\B) \supseteq \desc(E)$. If
$\Lang{}(\B)\cap Bad = \emptyset$ then this proves that $\desc(E)\cap
Bad=\emptyset$, and thus that none of the ``bad'' configurations is
reachable. We now define tree automata.

Let $\Q$ be a finite set of symbols, with arity $0$, called {\em states} such
that $\Q \cap \F= \emptyset$.  $\TFQ$ is called the set of {\em configurations}.
\begin{definition}[Transition, normalized transition, $\varepsilon$-transition]
  \label{def:normalized}
  A {\em transition} is a rewrite rule $c \f q$, where $c$ is a
  configuration i.e. $c \in \TFQ$ and $q \in \Q$. A {\em normalized
    transition} is a transition $c \f q$ where $c = f(q_1, \ldots,
  q_n)$, $f \in \F$ whose arity is $n$, and $q_1, \ldots, q_n \in \Q$.
  An {\em $\varepsilon$-transition} is a transition of the form  $q \f q'$ where $q$ and $q'$ are states. 
\end{definition}

\begin{definition}[Bottom-up nondeterministic finite tree automaton]
  A bottom-up nondeterministic finite tree automaton (tree automaton for short)
  is a quadruple $\A= \langle \F, \Q, \Q_F,\Delta \cup \Deps \rangle$, where $\Q_F
  \subseteq \Q$, $\Delta$ is a set of normalized transitions
  and $\Deps$ is a set of $\varepsilon$-transitions.
\end{definition}

The {\em rewriting relation} on $\TFQ$ induced by the transitions of $\A$ (the
set $\Delta \cup \Deps$) is denoted by $\f_{\Delta\cup\Deps}$.  When $\Delta$ is
clear from the context, $\f_{\Delta\cup\Deps}$ will also be denoted by
$\f_{\A}$. We also introduce $\rwne_\A$ the \emph{transitive relation} which is induced by the set
$\Delta$ alone.


\begin{definition}[Recognized language, canonical term]
  The tree language recognized by $\A$ in a state $q$ is $\Lang{}(\A,q) = \{t \in \TF \sep t \f^*_{\A} q \}$.
  The language recognized by $\A$ is $\Lang{}(\A) = \bigcup_{q \in \Q_F} \Lang{}(\A, q)$. A tree language is regular if
  and only if it can be recognized by a tree automaton.
  A term $t$ is a {\em canonical term} of the state $q$, if $t \rwne_\A q$.
\end{definition}

\begin{example}
   Let $\A$ be the tree automaton $\langle \F, \Q, \Q_F, \Delta \rangle$ such
   that $\F=\{f,g,a\}$, $\Q= \{q_0, q_1, q_2\}$, $\Q_F=\{q_0\}$,  $\Delta= \{f(q_0)
   \rw q_0, g(q_1) \rw q_0, a \rw q_1, b \rw q_2 \}$ and
   $\Delta_{\epsilon}=\{q_2 \rw q_1 \}$. In $\Delta$, transitions are {\em
     normalized}. A transition of the form $f(g(q_1)) \f q_0$ is not
   normalized. The term $g(a)$ is a term of $\TFQ$ (and of $\TF$) and can be
   rewritten by $\Delta$ in the following way: $g(a) \rwne_\A g(q_1)
   \rwne_\A q_0$. Hence $g(a)$ is a canonical term of $q_1$. Note also that
   $b \rw_\A q_2 \rw_\A q_1$. Hence, $\Lang{}(\A, q_1)=
   \{a, b\}$ and $\Lang{}(\A)=\Lang{}(\A, 
   q_0) = \{g(a), g(b),f(g(a)), f(f(g(b))),\ldots\}=\{f^*(g([a|b]))\}$.
 \end{example}


\section{The Tree Automata Completion with $\varepsilon$-transitions}



Given a tree automaton $\A$ and a TRS $\R$, the tree automata completion
algorithm, proposed in~\cite{Genet-RTA98,FeuilladeGVTT-JAR04}, computes a \emph{tree complete
automaton} $\aaex^*$ such that $\Lang{}(\aaex^*)=\desc(\Lang{}(\A))$ when it is
possible (for some of the classes of TRSs where an exact computation is
possible, see~\cite{FeuilladeGVTT-JAR04}), and such that $\Lang{}(\aaex^*)
\supseteq \desc(\Lang{}(\A))$ otherwise. 
In this paper, we only consider the exact case.

The tree automata completion with $\varepsilon$-transtions works as follow.
From $\A=\aaex^0$ completion builds a sequence $\aaex^0.\aaex^1\ldots\aaex^k$ of automata such that if
$s\in\Lang{}(\aaex^i)$ and $s\f_{\R} t$ then $t\in\Lang{}(\aaex^{i+1})$. Transitions of $\aaex^i$ are denoted by the set
$\Delta^i \cup \Deps^i$. Since for every tree automaton, there exists a
deterministic tree automaton recognizing the same language, we can assume
that initially $A$ has the following properties:

\begin{property}[$\rwne$ deterministic]
  \label{prop:deterministic}
  If $\Delta$ contains two normalized transitions of the form 
  $f(q_1, \dots, q_n) \rw q$ and $f(q_1, \dots, q_n) \rw q'$, it means $q = q'$. 
  This ensures that the rewriting relation $\rwne$ is deterministic.
\end{property}

\begin{property}
  \label{prop:wellinitial}
  For all state $q$ there is at most one normalized transition $f(q_1, \dots, q_n) \rw q$
  in $\Delta$. This ensures that if we have $t \rwne q$ and $t' \rwne q$ then $t = t'$.
\end{property}

If we find a fixpoint automaton $\aaex^k$ such that $\desc(\Lang{}(\aaex^k)) =
\Lang{}(\aaex^k)$, then we note $\aaex^*=\aaex^k$ 
and we have $\Lang{}(\aaex^*) \supseteq \desc(\Lang{}(\aaex^0))$~\cite{FeuilladeGVTT-JAR04}.
To build $\aaex^{i+1}$ from $\aaex^{i}$, we achieve a \textit{completion step}
which consists of finding \textit{critical pairs} between $\f_{\R}$ and
$\f_{\aaex^i}$. To define the notion of critical pair, we extend the definition
of substitutions to the terms of $\TFQ$. For a substitution $\sigma:\X\mapsto\Q$ and
a rule $l\f r \in \R$, a critical pair is an instance $l\sigma $ of $l$ such
that there exists $q\in\Q$ satisfying $l\sigma \f^*_{\aaex^i}q$ and $l\sigma
\f_{\R} r\sigma$. Note that since
$\R$, $\aaex^i$ and the set $\Q$ of states of $\aaex^i$ are finite, there is only a finite
number of critical pairs. For every critical pair detected between $\R$ and
$\aaex^i$ such that we do not have a state $q$' for which $r\sigma \rwne_{\aaex^i}q'$ and $q' \rw q \in \Deps^i$, the
tree automaton $\aaex^{i+1}$ is constructed by adding new transitions $r\sigma \rwne q'$ to $\Delta^i$
and $q' \rw q$ to $\Deps^i$ such that $\aaex^{i+1}$ recognizes $r\sigma$ in $q$, i.e. $r\sigma \f^*_{\aaex^{i+1}} q$, see
Figure~\ref{fig:cp}.
\begin{figure}[!ht]
  {\small
    \[
    \xymatrix{
      l\sigma \ar[r]_-{\R}\ar[d]^-{*}_-{\aaex^i} & r\sigma \ar[d]_-{\not\varepsilon}^{\aaex^{i+1}}\\
      q & q' \ar[l]^-{\aaex^{i+1}}
    }
    \]}
  \vspace*{-7mm}
  \caption{\footnotesize A critical pair solved \label{fig:cp}
  }
\end{figure}
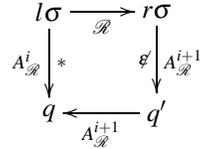
It is important to note that we consider the critical pair only if the
last step of the reduction $l\sigma \f^*_{\aaex^i}q$, is the last step of rewriting is not a $\varepsilon$-transition.
Without this condition, the completion computes the transitive closure of the
expected relation $\Deps$, and thus looses precision. 
The transition $r \sigma \f q'$ is not necessarily a normalized
transition of the form $f(q_1, \ldots, q_n) \f q'$ and so it has to be normalized
first. Instead of adding $r\sigma \rw q'$ we add $\norm(r\sigma \rw q')$ to
transitions of $\Delta^i$.
Here is the $\norm$ function used to normalize transitions. Note that, in 
this function, transitions are normalized using new states of $\Q_{new}$.

\begin{definition} [$\norm$] Let $\A=\la \F, \Q, \Qf, \Delta\cup\Deps\ra$ be a tree automaton, $\Q_{new}$ a
  set of {\em new} states such that $\Q\cap \Q_{new} = \emptyset$, $s \in \TFQ$ and $q'\in
  \Q$.
  The normalization of the transition $s \rw q'$ is done in two mutually inductive steps.
  The first step denoted by $\norm(s \rw q'\sep\Delta)$, we rewrite $s$ by $\Delta$ until rewriting 
  is impossible: we obtain a unique configuration $t$ if $\Delta$ respects the property~\ref{prop:deterministic}.
  The second step $\norm'$ is inductively defined by:
  \begin{itemize}
  \item
    $\norm'(f(t_1, \ldots, t_n) \rw q\sep\Delta)= \Delta \cup \{f(t_1, \ldots,
    t_n) \rw q\}$ if $\forall i = 1\ldots n:\ t_i \in \Q$
  \item 
    $\norm'(f(t_1, \ldots, t_n) \rw q \sep \Delta)= \norm(f(t_1, \ldots, q_i,\ldots, t_n) \rw q\sep \norm'(t_i\rw q_i\sep \Delta)\ )$
    where $t_i$ is subterm s.t. $t_i \in \TFQ\setminus \Q$ and $q_i \in \Q_{new}$.
  \end{itemize}
\end{definition}

 \begin{lemma}
   \label{lem:welldefined}
   If the property \ref{prop:deterministic} holds for $\aaex^i$ then it holds also for $\aaex^{i+1}$.
 \end{lemma}

 \begin{proof}[Intuition]
   The determinism of $\rwne$ is preserved by $\Delta$, since when a new set of transitions
   is added to $\Delta$ for a subterm $t_i$, we rewrite all other subterms $t_j$ with the new $\Delta$ until rewriting is impossible 
   before resuming the normalization. Then, if we try to add to $\Delta$ a transition $f(q_1, \dots, q_n) \rw q$
   though there exists a transition $f(q_1, \dots, q_n) \rw q'\in \Delta$, it means that the configuration $f(q_1, \dots, q_n)$ 
   can be rewritten by $\Delta$. This is a contradiction : when we resume the normalization all subterms $t_i$ can not be rewritten 
   by the current $\Delta$. So, we never add a such transition to $\Delta$. The normalization produces a new set of transitions $\Delta$
   that preserves the property \ref{prop:deterministic}.
 \end{proof}

It is very important to remark that the transition $q'\rw q$ in Figure~\ref{fig:cp}
creates an order between the language recognized by $q$ and the one recognized by
$q'$.  Intuitively, we know that for all substitution $\sigma' : \X \rw \TF$ such that $l\sigma'$ is
a term recognized by $q$, it is rewritten by $\R$ into a canonical term ($r\sigma'$) of $q'$.
By duality, the term $r\sigma'$ has a parent ($l\sigma'$) in the state $q$.
Extending this reasoning, $\Deps$ defines a relation between canonical
terms. This relation follows rewriting steps at the top position and forgets
rewriting in the subterms.

\begin{definition}[$\arw$]
  Let $\R$ be a TRS. For all terms $u$ $v$, we have $u \arw_{\R} v$ iff there exists
  $w$ such that $u \rw_\R^* w$, $w \trw_{\R} v$ and there is not
  rewriting on top position $\lambda$ on the sequence denoted by $u
  \rw_\R^* w$.
  
\end{definition}

In the following, we show that the completion builds a tree automaton where
the set $\Deps$ is an \emph{abstraction} $\arw_{\R_i}$ of the rewriting relation $\rw_\R$, for
any relevant set $\R_i$.

\begin{theorem}[Correctness]
  \label{thm:correct}
  Let be $\aaex^*$ a complete tree automaton 
  such that $q'\rw q$ is a $\varepsilon$-transition of $\aaex^*$. Then, for all canonical
  terms $u$ $v$ of states $q$ and $q'$ respectively s.t. $q'\rw q$, we have :
  \vspace*{-5mm}
  
  \[\xymatrix{ 
    u \ar[d]^-{\not\varepsilon}_-{\aaex^*} \ar@{-->}[r]_-{\R}
                &v \ar[d]^-{\not\varepsilon}_-{\aaex^*}\\ 
    q &q' \ar[l] }
  \]
\end{theorem}

First, we have to prove that the property \ref{prop:deterministic} is preserved by completion.
To prove theorem \ref{thm:correct}, we need a stronger lemma.

\begin{lemma}[]
  \label{lem:correct}
  Let be $\aaex^*$ a complete tree automaton, $q$ a state of $\aaex^*$ and $v\in\Lang{}(\aaex^*,q)$.
  Then, for all canonical term $u$ of $q$, we have $u \rw_\R^* v$. 
\end{lemma}

\begin{proof}[Proof sketch]
  
  The proof is done by induction on the number of  completion steps
  to reach the post-fixpoint $\aaex^*$ : we are going to show that
  if $\aaex^i$ respects the property of lemma~\ref{lem:correct},
  then $\aaex^{i+1}$ also does.
  
  The initial $\aaex^0$ respects the expected property~: we consider
  any state $q$ and a canonical term $t$ of $q$: since no completion
  step was done, $\aaex^0$ has no $\varepsilon$-transitions. It means
  that for all term $t'\rwne q$. Thanks to the property
  \ref{prop:wellinitial}, we have $t = t'$ and obviously $t \rw^*_\R
  t'$.

  Now, we consider the normalization of a transition of the form $r\sigma \rwne q'$
  such that $l\sigma \rw^*_{\aaex^i} q$ with $\Delta$ the ground transition set and $\Deps$ the $\varepsilon$-transition set of $\aaex^i$.
  We show that the property is true for all new states (including $q'$). 
  Then, in a second time, we will show that it is true for state $q$,
  if we add the second transition of completion: $q'\rw q$. 

  Let us focus on the normalization of $\norm'(r\sigma \rw q'\sep \Delta)$ where for
  any existing state $q$ and for all $u\ v \in \TF$ such that $v \rw_{\Delta\cup\Deps} q$ and $u \rw_{\Delta} q$, we have $u \rw_\R^* v$.
  We show that for all $t \in \TFQ$, if we have $\Delta' = \norm'(t \rw q'\sep \Delta)$, for all $u\ v \in \TF$ such that $v \rw_{\Delta'\cup\Deps} q'$ and $u \rw_{\Delta'} q$, we have $u \rw_\R^* v$. 
  The induction is done on the 
  number of symbols of $\F$ used to build $t$.

  First case $\norm'(t \rw q \sep \Delta)$ where $t = f(q_1,\dots, q_n)$ : we define $\Delta'$ by adding the transition $f(q_1, \dots, q_n) \rw q$
  to $\Delta$, where $q$ is a new state. Then, for all substitutions $\sigma' : \Q \mapsto \TF$ such that $t\sigma' \rw_{\Delta\cup\Deps} q$, and all 
  substitutions $\sigma'' : \Q \mapsto \TF$ such that $t\sigma'' \rw_{\Delta'} q$ we aim at proving that $t\sigma''
  \rw_\R^* t\sigma'$. Since each state $q_i$
  is already defined, using the hypothesis on $\Delta$ we deduce that $\sigma''(q_i) \rw^*_\R \sigma'(q_i)$. This implies that $t\sigma'' \rw_\R^* t\sigma'$, the property 
  also holds for $\Delta'$.

  Second case $\norm'(t \rw q \sep \Delta)$ where $t = f(t_1,\ldots,t_n)$: we select $t_i$ a subterm of $t$, obviously the number
  of symbols is strictly lower to the number of symbols of $t$.
  By induction, for the normalization of $\norm'(t_i \rw q_i\sep \Delta)$ we have a new 
  set $\Delta'$ that respects the expected property. Then, we normalize $t$ into $t' = f(t'_1, \dots, q_i, \dots, t'_n)$, 
  the term obtained after rewriting with $\Delta'$ thanks to $\norm$. Since $t_i \not\in \Q$, the number of
  symbols of $\F$ in $t' = f(t_1, \dots, q_i, \dots, t_n)$ is strictly smaller than the number of symbols of $ \F$ in $t$. Note 
  that rewriting $t'$ with $\Delta'$ can only decrease the number of symbols of $\F$ in $t'$.
  Since $t'$ has a decreasing number of symbols and $\Delta'$ respects the property we can deduce by induction
  that we have $\Delta'' = \norm'(t'\rw q\sep \Delta')$ such that for all $v \rw_{\Delta''\cup\Deps} q'$ and $u \rw_{\Delta''} q$, $u \rw_\R^* v$.
  
  So, we conclude that the normalization $\norm'(r\sigma \rw q'\sep \Delta)$ computes $\Delta'$ the set of ground transitions for $\aaex^{i+1}$.
  For all terms $u$ $v$ such that $u \rw_{\Delta'\cup\Deps} q'$ and $u \rw_{\Delta'} q'$ we have $u \rw_\R^* v$. 

  Now, let us consider the second added transition $q' \rw q$ to $\Deps$, all canonical terms
  $r\sigma''$ of $q'$, and all terms $l\sigma''' \in
  \Lang{}(\aaex^i, q)$ such that $l\sigma''' \rw_\R r\sigma'''$ and
  $r\sigma''' = r\sigma''$.  By hypothesis on $\aaex^i$, we know that every canonical term $u$ of $q$
  we have $u \rw_\R^*
  l\sigma'''$. By transitivity, we have $u \rw_\R^* r\sigma''$.  The
  last step consists in proving that for all terms of all states of
  $\aaex^{i+1}$, the property holds: this can be done by induction on
  the depth of the recognized terms.
\end{proof}

The theorem \ref{thm:correct} is shown by considering the introduction of the
transition $q' \rw q$. By construction, there exists a substitution $\sigma : \X \mapsto \Q$ and a rule
$l \rw r \in \R$ such that we have $l\sigma \rw^*_{\aaex^*} q$ and $r\sigma \rwne_{\aaex^*} q'$. We consider all substitution  
$\sigma' : \X \mapsto \TF$ such that for each variable $x \in \vars(l)$, $\sigma'(x)$ is a canonical term
of the state $\sigma(x)$. Obviously, using the result of the lemma \ref{lem:correct},
for all canonical term $u$ of $q$ we have $u \rw^*_\R l\sigma'$. Since the last step of rewriting 
in the reduction $l\sigma \rw^*_{\aaex^*} q$ is not a $\varepsilon$-transition, we also deduce that $l\sigma'$ is not produced
by a rewriting at the top position of $u$ whereas it is the case for $r\sigma'$ and we have $u \arw_{\R} r\sigma'$.

\begin{theorem}[Completeness]
  Let $\aaex^*$ be a complete tree automaton, 
  $q,q'$ states of $\aaex^*$ and $u,v \in \TF$ such that $u$ is a canonical term of $q$
  and $v$ is a canonical term of $q'$. If $u \arw_\R v$ then there exists a $\varepsilon$-transition $q' \rw q$ in $\aaex^*$.
\end{theorem}
\begin{proof}[Proof sketch]
  By definition of $u \arw_\R v$ there exists a term $w$ such that $u \rw_\R^* w$ and
  and there exists a rule $l \rw r \in \R$ and a substitution $\sigma : \X \mapsto \TF$ such that 
  $w = l\sigma$ and $v = r\sigma$.
  Since $\aaex^*$ is a complete tree automaton, it is closed by rewriting. This means 
  that any term obtained by rewriting any term of $\Lang{}(\aaex^*, q)$ is also in $\Lang{}(\aaex^*, q)$. This
  property is true in particular for the terms $u$ and $w$. 
  Since $w$ is rewritten in $q$ by transitions of $\aaex^*$, we can define
  a second substitution $\sigma' : \X \mapsto \Q$ such that $l\sigma \rw^*_{\aaex^*} l\sigma' \rw^*_{\aaex^*} q$.
  Using again the closure property of $\aaex^*$, we know that the critical pair $l\sigma' \rw_\R r\sigma'$
  and $l\sigma' \rw^*_{\aaex^*} q$ is solved by adding the transitions $r\sigma' \rwne_{\aaex^*} q''$ and $q'' \rw q$. Since the property \ref{prop:deterministic}
  is preserved by completion steps, we can deduce that $q'' = q'$ which means $q' \rw q$.
\end{proof}

\begin{example}
  \label{the_example}
  To illustrate this result, we give a completed tree automaton for a small
  TRS. We define $\R$ as the union of the two sets of rules
  $\R_1 = \{ a \rw b,\; b\rw c \}$ and $\R_2 = \{f(c) \rw g(a),\; g(c) \rw h(a),\; h(c) \rw f(a)\}$. We define
  initial set $E=\{f(a)\}$.
  We obtain the following tree automaton fixpoint :
  {\small
  \[\aaex^*= \left\la
  \Qf = \{q_f\},\quad
  \Delta = \left\{
    \begin{array}{rcl}
      a & \rw & q_a \\
      b & \rw & q_b \\
      c & \rw & q_c \\
      f(q_a) & \rw & q_f \\
      g(q_a) & \rw & q_g \\
      h(q_a) & \rw & q_h \\
    \end{array}
  \right\}\:
  \Deps= \left\{
    \begin{array}{rcl}
      q_b & \rw & q_a \\
      q_c & \rw & q_b \\
      q_g & \rw & q_f \\
      q_h & \rw & q_g \\
      q_f & \rw & q_h \\
    \end{array}\right\}
  \:
\right\ra
  \]
  }
  
  If we consider the transition $q_h \rw q_g$, and its canonical terms $h(a)$
  and $g(a)$ respectively, we can deduce $g(a) \arw_\R h(a)$. This is obviously an
  abstraction since we have $g(a) \rw_\R^1 g(b) \rw_\R^1 g(c) \rw_\R^\lambda h(a)$.
\end{example}

In the following, we use the notation $\arw_{\R_i}$ to specify the
relation for a relevant subset $\R_i$ of $\R$. For instance, $u
\arw_{\R_i} v$ 
denotes that there exists $w$ such that $u \rw_\R^* w$ with no rewriting
at the $\lambda$ position of $u$ and $w
\rw_{\R_i}^\lambda v$. In example~\ref{the_example},
we can say that $g(a)
\arw_{\R_2} h(a)$.

\section{From Tree Automaton to Kripke Structure}
Let $\aaex^*= \la \T(\F), \Q, \Q_F, \Delta \cup \Deps \ra$ be a complete tree
automaton, for a given TRS $\R$ and an initial language recognized by $\A$. A
Kripke structure is a four tuple $\K = (S, S_0, R, L)$ where $S$ is a set of
states, $S_0 \subseteq S$ initial states, $R \subseteq S \times S$ a left-total
transition relation and $L$ a function that labels each state with a set of
predicates which are true in that state. In our case, the set of true predicates
is a regular set of terms.

\begin{definition}[Labelling Function]
  Let $\A_P = \la \T(\F), \Q, \Delta \ra$ be the structure defined from
  $\aaex^*$ by removing $\varepsilon$-transitions and final states.
  We define the labelling function $L : q \mapsto \la \TF, \Q, \{q\}, \Delta\ra$ as the function which associates
  to a state $q$ the automaton $\A_P$ where $q$ is the unique final state. We obviously have
  the property for all state state $q$ :
  \[\forall t \in \Lang{}(L(q)), \quad t \rwne_{\aaex^*} q\]
\end{definition}
Now, we can build the Kripke structure for the subset $\R_i$ of $\R$ on which
we want to prove some temporal properties.

\begin{definition}[Construction of a Kripke Structure]
  We build the 4-tuple $(S, S_0, R, L)$ from a tree automaton such
  that we have $S = Q$, $S_0 \subseteq S$ is a set of initial states,
  $R(q, q')$ if $q' \rw q \in \Deps$ and the labelling function $L$ as
  just defined previously.
\end{definition}

Kripke structures must have a complete relation $R$. For any state $q$
whose have no successor by $R$, we had a loop such that $R(q, q)$ holds. Note
that this is a classical transformation of Kripke structures~\cite{ClarkeGP}.
A Kripke structure is parametrized by the set $S_0$. It defines which connected 
component of $R$ we are interested to analyze. For instance, to analyze 
the abstract rewriting at the top position of terms in $\Lang{}(\aaex^*)$, we define
set $S_0 = \Q_F$ (the set of final states of $\aaex^*$), since all canonical
terms of final states are initial terms. 
For all abstract rewriting at a deeper position $p$, we need to define 
a set $Sub$ of initial subterms considered as the beginning of the rewriting
at the position $p$. Then the set $S_0$ will be defined as 
$S_0 = \{q \sep \exists t \in Sub,\; t \rwne_{\aaex^*} q\}$.

Kripke structure models exactly the abstract rewriting
relation $\arw_{\R_i}^*$ for the corresponding subset $\R_i \subseteq \R$.

\begin{theorem}
  Le be $\K=(S, S_0, R, L)$ a Kripke structure built from $\aaex^*$.
  For any states $s$, $s'$ such that $R(s, s')$ holds, there exists two
  terms $u \in L(s)$ and $v \in L(s')$ such that $u \arw_{\R_i} v$.
\end{theorem}

\begin{proof}
  Here, the proof is quite trivial. It is a consequence of the theorem \ref{thm:correct} which can be
  applied on the relation $R$ of the Kripke structure.
\end{proof}

In Example~\ref{the_example}, if we want to verify properties of $\R_1$
or $\R_2$, we need to consider a different subset of $\Deps$ corresponding
to the abstraction of the relation rewriting $\arw_{\R_i}$.
Figures~\ref{fig2}~and~\ref{fig3} show the Kripke structures corresponding to those
abstractions. Note that in figure~\ref{fig2}, a loop is needed on state $c$ to have a total relation
for $\K_1$.

\begin{figure}[!ht]
  \begin{minipage}{0.5\linewidth}
    \centering
    \begin{tikzpicture}[thick, initial text=]
      \tikzstyle{every node}=[font=\tiny]
      \tikzstyle{every state}=[minimum size=.8cm]
      \tikzstyle{accepting}=[accepting by double]
      \node [initial,state] (a) at (0, 0) {$q_a$}; 
      \node [state] (b) at (2, 0) {$q_b$};
      \node [state] (c) at (4, 0) {$q_c$};
      \node [] (f) at (0, 1.8) {};
      \draw[->] (a) edge (b) (b) edge (c) (c) [loop above] edge (c);
    \end{tikzpicture}
    \caption{\label{fig2}\footnotesize Kripke structure $\K_1$ for $\arw_{\R_1}$}
  \end{minipage}
  \begin{minipage}{0.5\linewidth}
    \begin{center}
      \begin{tikzpicture}[thick, initial text=]
        \tikzstyle{every node}=[font=\tiny]
        \tikzstyle{every state}=[minimum size=.1cm]
        \tikzstyle{accepting}=[accepting by double]
        \node [initial,state] (a) at (0, 0) {$q_f$}; 
        \node [state] (b) at (2, 0) {$q_g$}; 
        \node [state] (c) at (1, 1.5) {$q_h$}; 
        \draw[->] (a) edge (b) (b) edge (c) (c) edge (a);
      \end{tikzpicture}
      \caption{\label{fig3}\footnotesize Kripke structure $\K_2$ for $\arw_{\R_2}$}
    \end{center}
  \end{minipage}
\end{figure}

The set $S_0$ of initial states depends of the abstract rewriting relation selected.
For example, if we want to analyze $\arw_{\R_2}$ (or $\arw_{\R_1}$), we define $S_0=\{q_f\}$ (resp. 
$S_0 = \{q_a\}$).

\section{Verification of R-LTL properties}
To express our properties, we propose to define the Regular Linear
Temporal Logic (R-LTL). R-LTL is LTL where predicates are replaced by a tree
automaton. The language of such a tree automaton
characterizes a set of admissible terms. A state $q$ of a Kripke
structure validates the atomic property $P$ characterized by a tree automaton $\A_P$
if and only if one term recognized by $L(q)$ must be recognized by $\A_P$ to satisfy the
property. More formally:

\[\K(Q,\ Q_F,\ R,\ L),\ q \models P\quad \equ\quad \Lang{}(L(q)) \cap \Lang{}(\A_P) \neq \emptyset\]

We also add the operators ($\land$, $\lor$, $\neg$, $\nxt$, $\fut$, $\gbl$, $\unt$, $\rel$)
with their standard semantics as in LTL to keep the expressiveness
of the temporal logic. More information about these operators can
be found in~\cite{ClarkeGP}. Note that temporal properties do not range over the 
rewriting relation $\rw_\R$ but over its abstraction $\arw_\R$.
It means that the semantics of the temporal operators has to be interpreted
w.r.t. this specific relation. For example, the formula $\gbl(\{f(a)\} \imp \nxt \{g(a)\})$
on $\K_2$ (for more clarity, we note predicates as sets of terms): the formula 
has to be interpreted as : for all $q$ $q'$, if $\K_2,\ q \models \{f(a)\}$ and $R(q, q')$ then
we have  $\K_2,\ q' \models \{g(a)\}$. In the rewriting interpretation the only term $u$ such
that $f(a) \arw_{\R_2} u$ is $u = g(a)$.

We use the Büchi automata framework to perform model checking. A survey of this
technique can be found in the chapter 9 of~\cite{ClarkeGP}.  LTL (or R-LTL)
formulas and Kripke structures can be translated into Büchi automata. We
construct two Büchi automata : $\B_\K$ obtained from the Kripke structure and
$\B_L$ defined by the LTL formula. Since the set of behaviors of the Kripke
structure is the language of the automaton $\B_\K$, the Kripke structure
satisfies the R-LTL formula if all its behaviors are recognized by the automaton
$\B_L$. It means checking $\Lang{}(\B_\K) \subseteq \Lang{}(\B_L)$. For this
purpose, we construct the automaton $\overline{\B_L}$ that recognizes the
language $\overline{\Lang{}(\B_L)}$ and we check the emptiness of the automaton
$\B_\cap$ that accepts the intersection of languages $\Lang{}(\B_K)$ and
$\overline{\Lang{}(\B_L)}$. If this intersection is empty, the term rewriting
system satisfies the property. This is the standard model-checking technique.

$\B_\M$ and $\B_\K$ are classically defined as 5-tuples: alphabet, states,
initial states, final states and transition relation.
Generally, the alphabet of Büchi automata is a set of predicates.
Since we use here tree automata to define predicates, the alphabet of
$\B_\K$ and $\B_L$ is $\Sigma$ the set of tree automata that can be defined over $\TF$. 
Actually, a set of behaviors is a word
 which describes a sequence of states: if $\pi=s_0s_1s_2s_3\dots$ denotes
a valid sequence of states in the Kripke structure, then the word
$\pi' = L(s_0)L(s_1)L(s_2)\dots$ is recognized by $\B_\K$. The algorithms
used to build $\B_\M$ and $\B_\K$ can be found in~\cite{ClarkeGP}.

The automaton intersection $\B_\cap$ is obtained by computing the product of $\B_\K$ by $\overline{\B_L}$.
By construction all states of $\B_\K$ have to be final. Intuitively any infinite path
over the Kripke structure must be recognized by $\B_\K$. This case allows to use a 
simpler version of the general Büchi automata product.
\begin{definition}[$\B_\K \times \overline{\B_L}$]
  The product of $\B_\K = \la \Sigma,\; \Q,\; Q_i,\; \Delta,\; \Q\ra$ by $\overline{\B_L} = \la\Sigma,\; \Q',\;\Q'_i,\; \Delta',\; F\ra$ is defined as
  \[\la \Sigma,\; \Q \times \Q',\; \Q_i \times \Q'_i,\; \Delta_\times,\;  \Q \times F \ra \]
  where $\Delta_\times$ is the set of transitions $(q_\K, q_L)
  \stackrel{(\A_\K,\A_L)}{\lrw} (q'_\K, q'_L)$ such that $q_\K
  \stackrel{\A_\K}{\lrw} q'_\K$ is a
  transition of $\B_\K$ and $q_L
  \stackrel{\A_L}{\lrw} q'_L$ is a transition of $\overline{\B_L}$. Moreover, the transition 
  is only valid if the intersection between the languages of $\A_\K$ and $\A_L$ is non
  empty as expected by the satisfiability of the R-LTL atomic formula.
\end{definition}

Finally the emptiness of the language $\Lang{}(\B_\cap)$ can be checked using the standard algorithm
based on depth first search to check if final states are reachable.
\begin{example}
To illustrate the approach, we propose to check the formula $P =
\gbl(\{f(a)\} \imp \nxt \{g(a)\})$
on example~\ref{the_example}. The automaton $\overline{\B_L}$ (fig.~\ref{fig4}) recognizes the negation of the formula $P$
expressed as $\fut(\{f(a)\} \land \nxt \neg\{g(a)\})$ and $\B_\K$ (fig.~\ref{fig5}) recognizes the all behaviors of the Kripke structure
$\K_2$~(fig.~\ref{fig3}). The notation $\A_\alpha$ denotes the tree automaton such that its language
is described by $\alpha$ ($\A_{\neg g(a)}$ recognizes the complement of the language $\Lang{}(\A_{g(a)})$ and $\A_*$ recognizes all
term in $\TF$). Figure~\ref{fig6} shows the result of intersection $\B_\cap$ between $\B_\K$ and $\overline{\B_L}$. Only reachable
states and valid transitions (labeled by non empty tree automata intersection) are showed. Since no 
reachable states of $\B_\cap$ are final, its language is empty. It means that all behaviors of $\K_2$ satisfy $P$ : the
only successor of $f(a)$ for the relation $\arw_{\R_2}$ is $g(a)$.
\begin{figure}[!ht]
  \begin{minipage}{0.5\linewidth}
    \centering
    \begin{tikzpicture}[scale=.8,thick,initial text=]
      \tikzstyle{every node}=[font=\tiny]
      \tikzstyle{every state}=[minimum size=.1cm]
      \tikzstyle{accepting}=[accepting by double]
      \node [initial,state] (q1) at (0, 0) {$1$}; 
      \node [state] (q2) at (2, 0) {$2$};
      \node [accepting, state] (q3) at (4, 0) {$3$};
      \path[->]
      (q1)  edge [loop above] node {$\A_*$} (q1) 
      edge node [above] {$\A_{f(a)}$} (q2)
      (q2)  edge node [above] {$\overline{\A_{g(a)}}$} (q3)
      (q3)  edge [loop above] node {$\A_*$} (q3);
    \end{tikzpicture}
    \caption{\footnotesize Automaton $\overline{\B_L}$}
    \label{fig4}

    \begin{tikzpicture}[scale=.8,thick,initial text=]
      \tikzstyle{every node}=[node distance=40,font=\tiny]
      \tikzstyle{accepting}=[accepting by double]
      \tikzstyle{every state}=[accepting, minimum size=.1cm]
      \node [initial,state] (q4)            {$4$}; 
      \node [state] (q5) [right of=q4]      {$5$};
      \node [state] (q6) [below of=q5]      {$6$};
      \node [state] (q7) [right of=q6]      {$7$};
      \path[->]
      (q4)  edge []                         node [above]    {$L(q_f)$} (q5) 
      (q5)  edge []                         node [left]     {$L(q_g)$} (q6)
      (q6)  edge []                         node [below]    {$L(q_g)$} (q7)
      (q7)  edge [bend angle=40,bend right] node [right]    {$L(q_g)$} (q5);
    \end{tikzpicture}
    \caption{\footnotesize Automaton $\B_{\K}$}
    \label{fig5}
  \end{minipage}
  \begin{minipage}{0.5\linewidth}
    \centering
    \vspace{10mm}
    \begin{tikzpicture}[scale=.8,thick,initial text=,bend angle=40]
      \tikzstyle{every node}=[node distance=50,font=\tiny]
      \tikzstyle{accepting}=[accepting by double]
      \tikzstyle{every state}=[minimum size=.1cm]
      \node [initial,state] (q14)            {$1,4$}; 
      \node [state] (q15) [right of=q14, node distance=60]     {$1,5$};
      \node [state] (q16) [below of=q15]     {$1,6$};
      \node [state] (q17) [right of=q16, node distance=60]     {$1,7$};
      \node [state] (q25) [below of=q14]     {$2,5$};
      \path[->]
      (q14) edge []                 node [above]           {$\A_* \cap L(q_f)$}       (q15)
            edge []                 node [left]            {$\A_{f(a)}\cap L(q_f)$}    (q25)
      (q15) edge []                 node [left]            {$\A_* \cap L(q_g)$}       (q16)
      (q16) edge []                 node [above]           {$\A_* \cap L(q_h)$}       (q17)
      (q17) edge [bend right]       node [right]           {$\A_* \cap L(q_f)$}       (q15)
            edge [bend left]        node [below]           {$\A_{f(a)} \cap L(q_g)$}   (q25);
    \end{tikzpicture}
    \caption{\footnotesize Automaton $\B_\cap$}
    \label{fig6}
  \end{minipage}
\end{figure}

\end{example}

\section{Conclusion, Discussion}

In this paper, we show how to improve the tree automata completion mechanism to
keep the ordering between reachable terms. This ordering was lost in the
original algorithm~\cite{FeuilladeGVTT-JAR04}. Another contribution is the
mechanism making it possible to prove LTL-like temporal properties on such
abstractions of sets of reachable terms. The work presented here only deals with
finite state systems and exact tree automata completion results. Future plans
are to extend this result so as to prove temporal properties on
over-approximations of infinite state systems. A similar objective has already
been tackled in~\cite{MeseguerPM-TCS08}. However, this was done in a pure
rewriting framework where abstractions are more heavily constrained than in tree
automata completion~\cite{FeuilladeGVTT-JAR04}. Hence, by extending LTL formula
checking on tree automata over-approximations, we hope to ease the verification
of temporal formula on infinite state systems.


\section*{Acknowledgements}

Many thanks to Axel Legay and Vlad Rusu for fruitful discussions on this
work and to anonymous referees for their comments.

\bibliographystyle{eptcs} 



\end{document}